\documentclass[a4paper,twocolumn,11pt,accepted=2021-03-17]{quantumarticle}
\pdfoutput=1
\usepackage[utf8]{inputenc}
\usepackage[english]{babel}
\usepackage[T1]{fontenc}
\usepackage{amsmath}

\usepackage{tikz}
\usepackage{lipsum}

\usepackage{amsopn}
\usepackage{amsmath}
\usepackage{amssymb}
\usepackage{amsthm}
\usepackage{graphicx}
\usepackage{color}
\usepackage{listings}
\usepackage{subcaption}
\usepackage{enumerate}
\usepackage{cite}
\usepackage{mathtools}
\usepackage{todonotes}
\usepackage[colorlinks]{hyperref}

\usepackage{tikz}
\usetikzlibrary{angles,quotes}
\newcommand{\ie}{{\emph{i.e. \/}}}

\newtheorem{remark}{Remark}

\DeclareMathOperator{\tr}{tr}
\DeclareMathOperator{\Tr}{Tr}

\DeclareMathOperator{\diag}{diag}

\newcommand{\N}{\ensuremath{\mathbb{N}}}

\newcommand{\C}{\ensuremath{\mathbb{C}}}

\newcommand{\ket}[1]{\ensuremath{|#1\rangle}}
\newcommand{\bra}[1]{\ensuremath{\langle#1|}}
\newcommand{\ketbra}[2]{\ensuremath{\ket{#1}\bra{#2}}}
\newcommand{\proj}[1]{\ensuremath{\ketbra{#1}{#1}}}
\newcommand{\braket}[2]{\ensuremath{\langle{#1}|{#2}\rangle}}

\newcommand{\1}{{\rm 1\hspace{-0.9mm}l}}

\newcommand{\Lrm}{\ensuremath{\mathrm{L}}}

\newcommand{\ee}{\ensuremath{\mathrm{e}}}

\newcommand{\ii}{\ensuremath{\mathrm{i}}}

\newcommand{\MM}{\mathcal{M}}
\newcommand{\NN}{\mathcal{N}}
\newcommand{\DD}{\mathcal{D}}
\newcommand{\TT}{\mathcal{T}}
\newcommand{\PP}{\mathcal{P}}

\newcommand{\UU}{\mathcal{U}}
\newcommand{\HH}{\mathcal{H}}
\newcommand{\DU}{\mathcal{DU}}

\newcommand{\diaguni}{\ensuremath{\mathcal{DU}}}

\newcommand{\eqdef}{\mathrel{:=}}

\newcommand{\poptent}{p_\mathrm{opt}}  % minmal error discrimination withent 
\newcommand{\pent}{p_\mathrm{u}} % unambigous discrimination with ent assitance 
\newcommand{\p}{\tilde{p}_\mathrm{u}}  % unambigous discrimination without ent 
\newtheorem{lemma}{Lemma}
\newtheorem{theorem}{Theorem}

\newtheorem{corollary}{Corollary}

\newtheorem*{rem*}{Remark}

\usepackage{caption}
\def\>{\rangle}
\def\<{\langle}
\begin{document}

\title[Multiple-shot and unambiguous 
 discrimination]{Multiple-shot and unambiguous 
 discrimination of von Neumann measurements}

\author{Zbigniew Pucha{\l}a}
\affiliation{Institute of Theoretical and Applied Informatics, Polish 
Academy
of Sciences, ul. Ba{\l}tycka 5, 44-100 Gliwice, Poland}
\affiliation{Faculty of Physics, Astronomy and Applied Computer Science,
Jagiellonian University, ul. {\L}ojasiewicza 11, 30-348 Krak{\'o}w, 
Poland}
\orcid{0000-0002-4739-0400}

\author{{\L}ukasz Pawela}
\affiliation{Institute of Theoretical and Applied Informatics, Polish 
Academy
of Sciences, ul. Ba{\l}tycka 5, 44-100 Gliwice, Poland}
\email{lpawela@iitis.pl}
\orcid{0000-0002-0476-7132}

\author{Aleksandra Krawiec}
\affiliation{Institute of Theoretical and Applied Informatics, Polish 
Academy
of Sciences, ul. Ba{\l}tycka 5, 44-100 Gliwice, Poland}
\affiliation{Institute of Mathematics, Silesian University of 
Technology, ul. Kaszubska 23, 44-100 Gliwice, Poland}
\orcid{0000-0001-8390-6569}

\author{Ryszard Kukulski}
\affiliation{Institute of Theoretical and Applied Informatics, Polish 
Academy
of Sciences, ul. Ba{\l}tycka 5, 44-100 Gliwice, Poland}
\affiliation{Institute of Mathematics, University of Silesia, ul. 
Bankowa 14, 
40-007 Katowice, Poland}
\orcid{0000-0002-9171-1734}

\author{Micha{\l} Oszmaniec}
\affiliation{Institute of Theoretical Physics and Astrophysics, 
National Quantum Information Centre, Faculty of Mathematics,
Physics and Informatics, University of Gda{\'n}sk, ul. Wita Stwosza 57, 
80-308 
Gda{\'n}sk, Poland}
\affiliation{Center for Theoretical Physics, Polish Academy of Sciences\\ Al. 
Lotnik\'ow 32/46, 02-668
Warszawa, Poland }
\orcid{0000-0002-4946-6835}

\maketitle

\begin{abstract}
We present an in-depth study of the problem of multiple-shot discrimination of 
von Neumann measurements in finite-dimensional Hilbert spaces. Specifically, we 
consider two scenarios: minimum error and unambiguous discrimination.   In the 
case of minimum error discrimination, we focus on discrimination of 
measurements with the assistance of entanglement. 
%Interestingly,  we prove that in this setting,
We provide an alternative proof of the fact that
all pairs of distinct von Neumann measurements can be distinguished perfectly
(i.e. with the unit success probability) using only a finite number of queries.
Moreover, we analytically find the minimal number of queries needed for perfect 
discrimination.
We also show that in this scenario querying the measurements \emph{in parallel}
gives the optimal strategy, and hence  any possible adaptive methods do not 
offer any advantage over the parallel scheme. In the unambiguous discrimination 
scenario, we give 
the general expressions for the optimal discrimination probabilities with and
without the assistance of entanglement. Finally, we show that  typical pairs of
Haar-random von Neumann measurements can be perfectly distinguished with only
two queries.
\end{abstract}

\section{Introduction}
With the recent technological progress, quantum information science is not
merely a collection of purely theoretical ideas anymore. Indeed, quantum
protocols of increasing degree of complexity are currently being implemented on
more and more complicated quantum devices~\cite{Carolan2015,Boixo2018} and are
expected to soon yield practical solutions to some real-world
problems~\cite{preskill2018quantum}. This situation motivates the need for 
certification 
and benchmarking of various building-blocks of quantum
devices~\cite{RandBench2011,Aolita2015,Wooton2018} (see~\cite{eisert2020quantum}
for a recent review). Discrimination or quantum hypothesis testing constitute
one of the paradigms for assessing the quality of parts of quantum
protocols~\cite{Chefles2000,Barnett09,Bergou2010,StructureStateDISC,
pirandola2019fundamental}. In this work, we present a comprehensive study of
various scenarios of discrimination of von Neumann measurements  on  a
finite-dimensional Hilbert space. Here, by von Neumann measurements, we
understand fine-grained projective measurements. 
Such measurements
are vital for most of the protocols appearing in quantum information
and quantum computing. In fact, even the most general quantum measurements
are typically implemented using projective measurements performed on
enlarged Hilbert space via the so-called Naimark construction 
\cite{watroustheory}.
%[HERE CITE PERES/ Wartrous BOOK FOR EXAMPLE]. 
Due to imperfections present in currently available quantum devices   
\cite{preskill2018quantum}
%[Cite Preskill NISQ paper]
and increased complexity associated with Naimark construction,  standard
projective measurements, such as the computational basis measurement, are
the most common measurements implemented currently on quantum
processors. It is important to point out that despite their relative
simplicity, the actual implementation of projective measurements on
quantum hardware is imperfect. For example, in a recent demonstration
of quantum computational supremacy (advantage) by the collaboration of
Google and UCBS 
\cite{arute2019quantum}
%[CITE GOOGLE PAPER] 
the researchers reported
single-qubit measurement errors that are of order of a few percents.
This motivates the interest in certification strategies tailored
specifically to  von-Neumann measurements.

The general problem of quantum
channel discrimination has attracted a lot of attention in recent years. One of
the first results was the study of discrimination of unitary
operators~\cite{acin2001statistical,duan2007entanglement}. Later, this has been
extended to various settings, such as multipartite unitary
operations~\cite{duan2008local} and the case of discrimination among more than 
two unitary channels~\cite{chiribella2013identification}. In the   
work~\cite{duan2009perfect}
the authors formulated necessary and sufficient conditions under which quantum 
channels can be perfectly discriminated. Further works investigated the
adaptive~\cite{harrow2010adaptive,cope2017adaptive,pirandola2017ultimate,krawiec2020discrimination}
and parallel~\cite{duan2016parallel} schemes for  discrimination of channels.
Finally, some asymptotic results on discrimination of typical quantum channels
in large dimensions were obtained in~\cite{nechita2018almost}. Discrimination of
quantum measurements, being a subset of quantum channels, is thus of particular
interest. Some of the earliest results on this topic involve condition on
perfect discrimination of two
measurements~\cite{ji2006identification,Ariano2001,Chiribella2008,puchala2018strategies,sedlak2014optimal}

We are interested in the following problem. Imagine we have an unknown
device hidden in a black box. We know it performs one of the two possible von
Neumann measurements, either $\PP_1$ or $\PP_2$. Generally, whenever a quantum
state is sent through  the box,  the box produces, with probabilities predicted
by quantum mechanics, classical labels corresponding to the measurement
outcomes. Our goal is to find schemes that attain the optimal success
probability for discrimination of measurements. The results contained in this
work concern the following two scenarios:

\emph{Minimum error discrimination---} In this setting, we are allowed to use 
the black box containing von Neumann measurement many times. Furthermore, we can
prepare any input state with an arbitrarily large quantum memory (i.e., we can
use ancillas of arbitrarily large dimension), and we can perform any channels
between usages of the black box. This allows us to implement both parallel (see
Fig.~\ref{fig:parallel_introduction}) as well as adaptive discrimination
strategies (see Fig.~\ref{fig:adaptive}). We focus on the case of
entanglement-assisted discrimination. Our main finding is that in the
multiple-shot scenario, adaptive strategies do not offer any advantage over
parallel queries. Moreover, we derive an explicit dependence  of the diamond
norm distance in the multiple-shot scenario as a function of the diamond norm in
the single-shot setting. As a consequence, we recover a known result
\cite{ji2006identification} stating that given sufficiently many queries, every
pair of different von Neumann measurements can be distinguished perfectly (i.e.
with zero error probability). 
%\michal{As a by-product of our considerations we show that
%in the case of perfect discrimination, the dimension of auxiliary system needed
%to implement the discrimination protocol is upper bounded by three.}

\begin{figure}[h!]
\centering\includegraphics{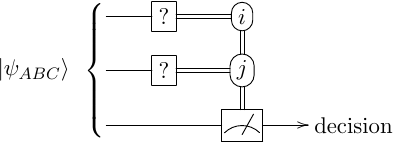}
\caption{A schematic representation of parallel discrimination scheme of 
quantum measurements that uses two applications of a measurement. }
\label{fig:parallel_introduction}
\end{figure}

\begin{figure}[h!]
\centering\includegraphics[scale=0.85]{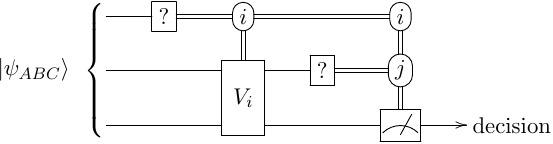}
\caption{A schematic representation of adaptive discrimination scheme 
of quantum measurements that uses two applications of a measurement. In this 
case adaptive scheme amounts to application of the unitary channel $\Phi_{V_i}$ 
which depends on the result $i$ of the first measurement.}
\label{fig:adaptive}
\end{figure}

\emph{Unambiguous discrimination---}  This scenario is an analogue  to the
well-known scheme of unambiguous state discrimination \cite{DIEKS1988}. Namely,
for every query to the black box, the decision procedure outputs  $\PP_1$,
$\PP_2$, or the inconclusive answer. The latter  means that the user cannot
decide which measurement was contained in the black box. Importantly, we require
that the procedure cannot wrongly identify the measurement (see
Fig.~\ref{fig:unam}). Our main contribution to this problem is the derivation of
the general schemes, which attain the optimal success probability both with and
without the assistance of entanglement. Interestingly, we find that optimal
success probability $p_u$ for unambiguous discrimination of projective
measurements $\PP_1, \PP_2$ with the assistance of entanglement is functionally
related to the diamond norm distance $\| \PP_1- \PP_2 \|_\diamond$, which
quantifies distinguishability of $\PP_1, \PP_2$ with the assistance of 
entanglement but without requiring unambiguous results. The specific formula is 
given by
\begin{equation}
p_u = 1 -\sqrt{1-\frac{1}{4} \left\Vert \PP_1- \PP_2 \right\Vert_\diamond^2}
\end{equation}
and its geometric interpretation is presented in Fig.~\ref{fig:circle-plot-0}.
Finally, we also present simple formulas for the optimal discrimination
probability of von Neumann measurements for qubits.

\begin{figure}[!h]
\centering
\begin{subfigure}{0.48\textwidth}
\includegraphics[width=\textwidth]{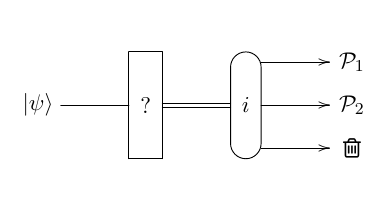}
\end{subfigure}
\begin{subfigure}{0.48\textwidth}
\includegraphics[width=\textwidth]{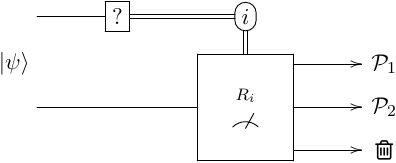}
\end{subfigure}
\caption{A schematic representation of the setting of unambiguous measurement 
discrimination scheme. The figure on the left shows an umbiguous dscrimination
without the assistance of entanglement while the figure on right shows
entanglement-assisted unambiguous discrimination.}\label{fig:unam}
\end{figure}

\begin{figure}
	\centering
	\includegraphics[width=0.48\linewidth]{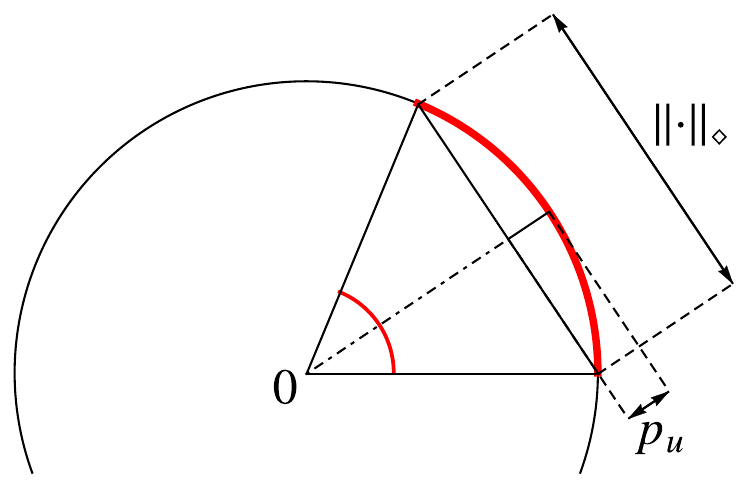}
	\caption{	
		Geometric interpretation of the relationship between the diamond norm 
		$\|\PP_1-\PP_2\|_\diamond$ and the probability of unambiguous 
		discrimination $\pent(\PP_{1},\PP_{2})$ presented on the unit circle.
	}
	\label{fig:circle-plot-0}
\end{figure}

\emph{Relation with classical channels---}
Let us contrast our results with transformations
mapping probability distributions to probability distributions, that is 
classical
channels. In this setting, we know that if such channels cannot be perfectly
distinguished in one shot, they cannot be distinguished perfectly in any finite
number of uses. What we can do is to study the asymptotic behavior of error
probability when the number of applications of the channel tends to infinity.
The error probability, formulated in the language of the hypothesis testing of 
two distributions, decays exponentially, and the optimal exponential error rate,
depending on a formulation, is given by the Stein bound, the Chernoff bound, the
Hoeffding bound, and the Han-Kobayashi bound,
see~\cite{hayashi2009discrimination} and references therein. In the case of
distinct von Neumann measurements and entanglement-assisted discrimination, we
show that one can perform perfect discrimination with the use of a finite number
of queries. Therefore, in contrary to the classical channels, we do not have to
consider the exponential error rate, as the error probability drops to zero in
a finite number of tries.

This work is organized as follows. In Section
\ref{sec:notation}, we give a survey of the main concepts and notation used
throughout this work (including the basic background on discrimination of
quantum channels and measurements). Then, in Section~\ref{sec:multSHOT}, we 
present our results for the scenario of
multiple-shot minimum error measurement discrimination. 
Theorem~\ref{th:parallel} therein expresses the minimum error in the 
parallel discrimination scheme as a function of the minimum error in the 
single-shot discrimination scheme.
Theorem~\ref{th:random} gives the upper bound on the probability of correct 
discrimination of a generic pair of von Neumann measurements coming from the 
Haar distribution.
The following Section~\ref{sec:unambiguousDISC} contains the results concerning 
the unambiguous discrimination of quantum measurements.
The main result of this section is formulated as Theorem~\ref{thm:ENTunamb}, 
which states the optimal success probability of unambiguous discrimination with 
the assistance of entanglement.
Lastly, in Section~\ref{sec:outro} we summarize our results and give some
directions for future research.

\section{Preliminaries and main concepts}\label{sec:notation}
By $\DD(\C^d)$ we will denote the set of quantum states on $\C^d$. Let
$\Lrm(\C^{d_1}, \C^{d_2})$ denote the set of all linear operators acting from
$\C^{d_1}$ to $\C^{d_2}$. For brevity we will put $\Lrm(\C^d, \C^d)\equiv
\Lrm(\C^d)$. A quantum channel is a linear mapping $\Phi: \Lrm(\C^{d_1}) \to
\Lrm(\C^{d_2})$ which is completely positive and trace-preserving. The former
means that for every $\Lrm(\C^{d_1}) \otimes \Lrm(\C^s)\ni \rho \geq 0$ we have
$\Lrm(C^{d_2}) \otimes \Lrm(\C^s) \ni (\Phi \otimes \1)(\rho) \geq 0$, while the
latter means $\Tr(\Phi(X)) = \Tr X$ for every $X \in \Lrm(\C^{d_1})$. A set of
generalized measurements (POVMs) $\C^d$ will be denoted by
$\mathrm{POVM}(\C^d)$.  A general quantum measurement $\MM$ on $\C^d$ is  a
tuple of positive semidefinite operators\footnote{In this paper we restrict our
attention to measurements with a finite number of outcomes.} on $\C^d$ that add
up to identity on $\C^d$ i.e. $\MM=(M_1,\ldots, M_n)$ with $M_i \geq 0$ and
$\sum_i M_i  =\1$.  If a quantum state $\sigma$ is measured by a measurement
$\MM$, then the outcome $i$ is obtained with the probability
$p(i|\sigma,\MM)=\tr (\sigma M_i)$ (Born rule). Therefore, a quantum measurement
$\MM$ can be uniquely  identified with a quantum channel
\begin{equation}\label{eq:measurePrepare}
\Psi_\MM (\sigma) =\sum_{i=1}^n \tr(M_i \sigma) \ketbra{i}{i},
\end{equation}
where states $\ketbra{i}{i}$ are perfectly distinguishable (orthogonal) pure
states that can be regarded as states describing the state of a classical
register. In what follows we will abuse the notation and simply treat quantum 
measurements (denoted by symbols $\MM,\NN,\PP$, ...) as quantum channels having 
the classical outputs. Using this interpretation one can readily use the results
concerning the discrimination of quantum channels for generalized measurements.
In particular, for entanglement-assisted discrimination of quantum channels we
have a classic result due to Helstrom~\cite{helstrom1976quantum}. It states
that the probability of correct discrimination $p_\mathrm{opt}(\Phi,\Psi)$ 
between two quantum channels $\Phi$ and
$\Psi$ is given by
\begin{equation}
p_\mathrm{opt}(\Phi,\Psi) = \frac{1}{2} + \frac{1}{4} \left\Vert \Phi - \Psi
\right\Vert_\diamond,
\end{equation}
where $\left\Vert \mathbf{S} \right\Vert_\diamond = \max_{\Vert X
\Vert_1 = 1} \left\Vert (\mathbf{S} \otimes \1 ) (X) \right\Vert_1$ denotes the
diamond norm of the linear map $\mathbf{S}$. Any optimal $X$ is called a 
discriminator. Thus, if the value of the diamond norm of the difference of two 
channels is strictly smaller than two, then the two channels cannot be 
distinguished perfectly in a single-shot scenario. 

In this work we will be concerned with von Neumann measurements i.e. projective
and fine-grained measurements on a Hilbert space of a given dimension $d$. Von
Neumann measurements $\PP$ in $\C^d$ are tuples of orthogonal projectors on
vectors forming an orthonormal basis $\lbrace \ket{\psi_i}\rbrace_{i=1}^d$ in
$\C^d$ i.e.
\begin{equation}\label{eq:projDEF}
\PP= \left( \ketbra{\psi_1}{\psi_1},\ketbra{\psi_2}{\psi_2},\ldots,  
\ketbra{\psi_d}{\psi_d} \right).
\end{equation}
In  what follows we will use $\PP_{\1}$ to denote the measurement in the
standard computational basis. We will also use $\PP_U$ to denote the von Neumann
measurement in the basis $\ket{\psi_i}=U \ket{i}$ for a unitary $d\times d$
matrix $U\in \mathcal{U}_d$. 
In other words, vectors $\ket{\psi_i}$ from Eq.\eqref{eq:projDEF} are columns 
of the matrix $U$. 
We also specify the subset $\DU_d \subset \UU_d$ of
diagonal, unitary matrices. 
Consider now a general task
of discriminating between two projective  measurements $\PP_{U_1}$, $\PP_{U_2}$,
and let $p_\mathrm{opt}(\PP_{U_1},\PP_{U_2})$ be the optimal probability for
discriminating between these measurements (both for minimum error and for 
unambiguous discrimination). Then, due to the unitary invariance of the
discrimination problem and identity $\PP_U (\cdot)=\PP_{\1} \circ (U^\dagger
\cdot U)$ we obtain that
$p_\mathrm{opt}(\PP_{U_1},\PP_{U_2})=p_\mathrm{opt}(\PP_\1,\PP_{U_{1}^\dagger
U_2}) $.  Therefore, for any discrimination of projective measurements, without
loss of generality, we can limit ourselves to considering the problem of
distinguishing between the measurement in the standard basis $\PP_\1$ and
another projective measurement $\PP_U$. It is important to note that definition
of measurement $\PP_U$  distinguishes projective measurements differing only by
ordering of elements of the basis. Moreover, a set of unitary matrices $\lbrace
UE|\ E\in \DU_d \rbrace$ specifies the same projective measurement, ie.
$\PP_U =\PP_{UE}$ for all $E \in \DU_d$.

Distinguishability of quantum measurements is strictly related to the
distinguishability of unitary channels. The prominent
result~\cite{chiribella2009theoretical,bisio2011quantum} gives an expression
which makes calculating the diamond norm of the difference of unitary channels
$\Phi_U$, $\Phi_\1$ substantially easier.  It says that the for a unitary matrix
$U$ we have
\begin{equation} \label{eq:dimondNORMuni}
\left\Vert \Phi_U - \Phi_\1 \right\Vert_\diamond = 2 \sqrt{1 - \nu^2},
\end{equation}
where $\nu  =  \min_{x \in W(U)} \vert x \vert $ and 
$W(X) \coloneqq \{ \bra{\psi} X \ket{\psi}: \braket{\psi}{\psi} = 1 \}$ 
denotes the
numerical range of the operator $X$. Building on this result, in
\cite{puchala2018strategies} the following characterization of the diamond
norm of the distance
between von Neumann measurements was obtained
\begin{equation}\label{eq:mainRESold}
\left\Vert \PP_U - \PP_\1 \right\Vert_\diamond = \min_{E \in \mathcal{DU}_d} 
\left\Vert \Phi_{UE} - \Phi_\1 \right\Vert_\diamond. 
\end{equation}
Therefore the distance between two von Neumann measurements is the minimal value
of the diamond norm of the difference between optimally coherified
channels~\cite{korzekwa2018coherifying}. Moreover, there is a simple condition
which lets us check whether two measurements are perfectly distinguishable
\cite{puchala2018strategies,duan2016parallel,ji2006identification}. It holds
that $\left\Vert \PP_U - \PP_\1 \right\Vert_\diamond =2$ if and only if there
exists a state $\rho$ such that
\begin{equation}\label{eq:perfect_distinguishability}
\diag \left( \rho U  \right) = 0.
\end{equation}

Equation \eqref{eq:mainRESold} will be of significant importance throughout 
this work.
We will also make use of the dephasing channel denoted   
$\Delta(\rho) = \sum_i \proj{i} \rho \proj{i}$.
Finally, when talking about the eigenvalues of a unitary matrix $U\in \UU_d$, 
we will follow convention that $\lambda_1=\ee^{\ii 
\alpha_1},\ldots,\lambda_d=\ee^{\ii \alpha_d}$ 
are ordered by 
their phases, if $\alpha_1 \le \ldots \le \alpha_d$ for
$\alpha_1,\ldots,\alpha_d \in [0,2 \pi)$. The angle of the shortest arc
containing all eigenvalues will be denoted by $\Theta (U)$.

\section{Minimum error discrimination}\label{sec:multSHOT}
Multiple-shot minimum error discrimination is a natural generalization of the 
single-shot
scheme studied in \cite{puchala2018strategies}. In this scenario we have
access to multiple queries to quantum measurements which is mathematically
equivalent to the problem of single-shot discrimination of channels
$\PP_{U^{\otimes N}}, \PP_{\1^{\otimes N}}$ acting on states defined on
multiple-component Hilbert space $(\C^d)^{\otimes N}$ (see Fig.
\ref{fig:parallel_scheme}). 
%For the sake of simplicity, i
In the following 
sections will write shortly $\PP_{\1}$ instead of $\PP_{\1^{\otimes N}}$.

\begin{figure}
\includegraphics[scale=1]{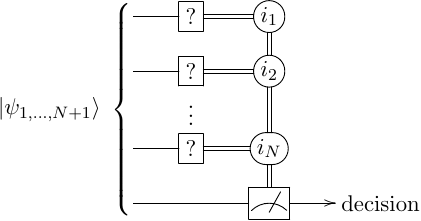}
\caption{A schematic representation of the parallel discrimination scheme of 
quantum measurements. }
 \label{fig:parallel_scheme}
\end{figure}

As we described in the preceding section, the problem of distinguishing
quantum measurements is intimately related to distinguishing unitary  channels
\cite{puchala2018strategies}. In what follows we leverage this result to prove a
number of results regarding multiple-shot discrimination of von Neumann
measurements.

\subsection{Optimality of the parallel scheme}\label{sec:optim-parallel}
The authors of~\cite{chiribella2008memory} showed that
parallel discrimination scheme is optimal among all possible architectures for
the case of discrimination of unitary channels.
In this subsection we will prove a theorem which thesis is rendered in the
spirit of the results obtained in~\cite{chiribella2008memory},
nevertheless, our theorem concerns the discrimination of 
von Neumann measurements. 
It is worth mentioning here that the optimality of the parallel scheme is no 
longer true when studying the discrimination of general (non-projective) POVMs. 
In the case of 
discrimination of POVMs with rank-one effects one may need to use adaptive 
discrimination scheme to obtain perfect distinguishability 
\cite{krawiec2020discrimination}.
While considering the discrimination of von Neumann measurements, one could
expect that we can improve the discrimination by performing some processing
based on the obtained measurement's labels. It appears that such processing will
not improve the discrimination.

We will focus on the minimum error discrimination in the 
parallel scenario and our goal will be to characterize the probability of 
correct discrimination after $N$ queries to the black box.
The first step is to extend Eq.~\eqref{eq:mainRESold} to the parallel 
setting. We   study the form of the optimal matrix $E$ in the parallel 
scheme. The following theorem, which proof is presented in Appendix 
\ref{app:parallel}, states that it has a tensor 
product form.

\begin{theorem}\label{th:parallel}
Let $N \in \mathbb{N}$, $U \in U_d$ and let $\PP_U$ be the corresponding von 
Neumann 
measurement on
$\C^d$. Then we have the following equality
\begin{equation}
\left\Vert  \PP_{ U^{\otimes N}} - \PP_{\1} \right\Vert_\diamond = \min_{E \in
\mathcal{DU}_d} \left\Vert  \Phi_{U^{\otimes N} E^{\otimes N}} - \Phi_{\1}
\right\Vert_\diamond.
\end{equation} 
\end{theorem}

Now we are interested in calculating the number of usages of the black box
required for perfect discrimination. Let us recall here that in the case of
distinguishing unitary operations this can always be achieved in a finite number
of steps $N=\lceil\frac{\pi}{\Theta(U)}\rceil$~\cite{acin2001statistical}. A 
similar result is achievable in the case of distinguishing von Neumann
measurements. Let $UE_0$ be an \emph{optimal} unitary matrix, that is a matrix 
for which
\begin{equation}
\left\Vert  \PP_{U} - \PP_{\1} \right\Vert_\diamond=\left\Vert  
\Phi_{UE_0} - 
\Phi_{\1} \right\Vert_\diamond.
\end{equation} 
Now we compute the number of queries needed for perfect discrimination.
We have already proven in Theorem \ref{th:parallel} that if $UE_0$ is
the optimal matrix, then the matrix $(UE_0)^{\otimes N}$  will also be optimal.
Therefore, the value of diamond norm $\left\Vert
\PP_{U^{\otimes N}} - \PP_{\1} \right\Vert_\diamond$ can be expressed as a 
function of $\Theta \left( (UE_0)^{\otimes N}\right)$. 
As long as $0 \not\in W\left((UE_0)^{\otimes N}\right)$, it holds that
$\Theta((UE_0)^{\otimes N})=N\Theta(UE_0)$.
The first time zero enters the 
numerical range
$W\left((UE_0)^{\otimes N}\right)$ is therefore equal to $ N=\lceil
\frac{\pi}{\Upsilon (U)} \rceil$, where $\Upsilon(U)$ is an optimized version of
$\Theta(U)$ \ie

\begin{equation}\label{eq:palma}
\Upsilon (U) \coloneqq \min_{E
\in\mathcal{DU}_d} \Theta(UE).
\end{equation}

We summarize the above discussion as the Corollary below providing the value of 
$\left\Vert  \PP_{U^{\otimes N}} - \PP_{\1} \right\Vert_\diamond$ in terms of 
$\Theta(U)$.
\begin{corollary}\label{cor:finite-N}
Let $N \in \N$, $U \in \mathcal{U}_d$.
The following holds
\begin{enumerate}[(i)]
\item if $N \Upsilon(U) \geq \pi$, then  
$\Vert\PP_{U^{\otimes N}} - \PP_\1 \Vert_\diamond = 2$;
\item if $N \Upsilon(U) < \pi$, then  
$\Vert \PP_{U^{\otimes N}} - \PP_\1 \Vert_\diamond = 2 \sin \left(\frac{N}{2} 
\Upsilon(U)\right)$.
\end{enumerate}
\end{corollary}

Another interesting property resulting from Theorem~\ref{th:parallel} and its
proof is the amount of the discriminator's entanglement with environment. The 
minimal dimension of an auxiliary system needed for optimal
discrimination is equal to the rank of the input state. We found out that for 
the majority of von Neumann measurements it is sufficient when the dimension of 
the auxiliary system is either two or three. A more detailed discussion is 
presented in Appendix \ref{app:parallel} after the proof of 
Theorem~\ref{th:parallel}.

Eventually, we are in the position to prove the optimality of parallel
discrimination scheme, which we present as the following theorem.

\begin{theorem} \label{th:parallel-optimal}
Let $U\in \UU_d$.
Consider the distinguishability of general quantum network with $N$ uses of the 
black box in which 
there is one of two measurements - either $\PP_U$ or $\PP_\1$.
Then the probability of correct distinction cannot be better then
in the parallel scenario.
\end{theorem}

\begin{proof}
Without loss of generality we may assume that the processing is performed 
using
only unitary operations. Indeed, using Stinespring dilation theorem, any 
channel
might be represented via a unitary channel on a larger system followed by 
the
partial trace operation. What is left to observe is that $\left\Vert
\tr_{B}(X_{AB}) \right\Vert_1 \le \left\Vert X_{AB} \right\Vert_1$ for 
arbitrary
bipartite matrix $X_{AB}$.
\begin{figure*}[t]
	\includegraphics[width=\textwidth]{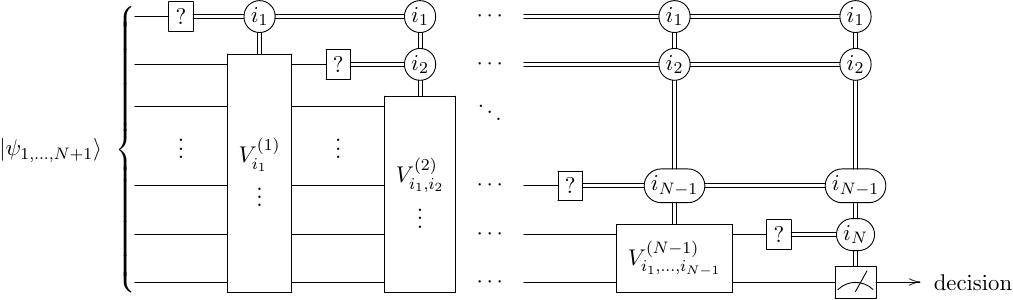}
	\caption{Schematic depiction of the sequential scheme. In the
		$k$\textsuperscript{th} step, after obtaining the label $i_k$, we 
		utilize all 
		labels $i_1,\ldots,i_k$ to modify the remaining parts of the state 
		in the hope 
		of improving distinguishability.}\label{fig:sequential_scheme}
\end{figure*}
	
The sequential scheme of discrimination of von Neumann measurements is shown in
Fig.~\ref{fig:sequential_scheme} and can be expressed as a channel 
\begin{equation}
\Psi_{U} = \left( \Delta_{1, \ldots , N} \otimes \1 \right) \Phi_{A_U},
\end{equation}
associated with a matrix $A_U$. Here $\Delta_{1, \ldots , N}$ is the dephasing
channel on subsystems $1,\ldots,N$. The channel $\Phi_{A_U}$ is shown in
Fig.~\ref{fig:au} and the exact form of this transformation can be found in
Appendix~\ref{app:au}.
	
\begin{figure*}[t]
	\includegraphics[width=\textwidth]{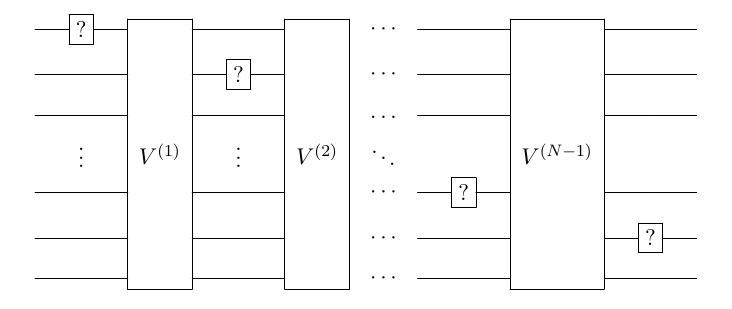}
	\caption{Visualization of the channel $\Phi_{A_U}$. Here $V^{(k)} = 
		\sum_{i_1,\ldots,i_k} \proj{i_1,\ldots,i_k} \otimes 
		V^{(k)}_{i_1,\ldots,i_k}$.}\label{fig:au}
\end{figure*}

Assuming that matrix $U$ is optimal, that is $\Upsilon (U)=\Theta(U)$, we have 
$\left\Vert \PP_{U^{\otimes N}} - \PP_\1 \right\Vert_\diamond = \left\Vert 
\Phi_{U^{\otimes N}} - \Phi_{\1}  \right\Vert_\diamond$. Hence, we 
may calculate the distance between $\Psi_U$ 
and 
$\Psi_\1$ as
\begin{equation}\label{eq: sequential}
\begin{split}
&\max_{\rho}  \left\Vert \left(\Psi_U - \Psi_\1\right)(\rho) \right\Vert_1\\
&= \max_{\rho} \left\Vert \left[ \left( \Delta_{1, \ldots , N} \otimes \1   
\right)     
\left( \Phi_{A_U} - \Phi_{A_\1} \right)\right] (\rho) \right\Vert_1 \\
&\leq \max_{\rho} \left\Vert \left(\Phi_{A_U} - \Phi_{A_\1}\right)(\rho) 
\right\Vert_1 \\
&\leq \max_{\rho} \left\Vert \left( \Phi_{U^{\otimes N} 
\otimes 
	\1} - \Phi_{\1} \right)(\rho)
\right\Vert_1\\
&= \left\Vert \Phi_{U^{\otimes N}} - \Phi_{\1}  \right\Vert_\diamond
= \left\Vert \PP_{U^{\otimes N}} - \PP_\1 \right\Vert_\diamond,
\end{split}
\end{equation}
where we maximize over states $\rho$ of appropriate dimensions. The induced 
trace norm is monotonically decreasing under the action of channels and 
this 
gives us first inequality. The second
one follows from the optimality of the parallel scheme of distinguishing
unitary channels~\cite{chiribella2008memory}. Therefore, the adaptive 
scenario does not give any advantage over the parallel scheme.
\end{proof}

\subsection{Discrimination of random measurements}
The structural characterization of multiple-shot discrimination of von Neumann
measurements given above allows us to draw strong conclusions about
distinguishability of  generic pairs of von Neumann measurements. In this work 
we restrict our attention to pairs of measurements distributed independently
according to the natural distribution coming from the Haar measure
$\mu(\UU_d)$~\cite{bengtsson2017geometry}.

\begin{theorem}\label{th:random} 
Consider two independently distributed Haar-random von Naumann
measurements on $\C^d$, i.e. $\PP_{U},\PP_V$, where  $U \sim \mu(\UU_d),\ V \sim
\mu(\UU_d)$. Let $\poptent(\PP_{U^{\otimes 2}},\PP_{V^{\otimes 2}})$ be the
optimal probability of discrimination measurements $\PP_U$ and $\PP_V$ using two
queries and assistance of entanglement. Then, we have the following bound
\begin{equation}
\underset{U,V \sim \mu(\UU_d)}{\mathrm{Pr}} \left( \poptent(\PP_{U^{\otimes
2}},\PP_{V^{\otimes 2}}) <1 \right)\leq \frac{1}{2^{d-1}}.
\end{equation}
In other words, in the limit of large dimensions $d$, typical Haar-random von 
Neumann measurements are perfectly distinguishable with the usage of two 
queries and assistance of entanglement (the probability that they cannot be 
perfectly distinguished is exponentially suppressed as a function of $d$).
\end{theorem}
\begin{proof}
From the unitary invariance of the Haar measure and the symmetry of the 
problem of measurement discrimination it follows that the distribution of  the 
random variable  $\poptent(\PP_{\1^{\otimes 2}},\PP_{U^{\otimes 2}})$ is 
identical to that of $\poptent(\PP_{U^{\otimes 2}},\PP_{V^{\otimes 2}})$. 
Consequently, we have
\begin{equation}\label{eq:equalityOFprob}
\begin{split}
&\underset{U,V \sim \mu(\UU_d)}{\mathrm{Pr}} \left( \poptent(\PP_{U^{\otimes 
2}},\PP_{V^{\otimes 2}})<1  \right)\\ & =\underset{U \sim 
\mu(\UU_d)}{\mathrm{Pr}} 
\left( \poptent(\PP_{\1^{\otimes 2}},\PP_{U^{\otimes 2}})<1  \right).
\end{split}
\end{equation}
From Corollary~\ref{cor:finite-N} it follows that the condition $\Vert \PP_{U} -
P_\1 \Vert_\diamond \geq \sqrt{2}$ implies  $\Vert \PP_{U^{\otimes 2}} -
P_{\1^{\otimes 2}} \Vert_\diamond =2$ and consequently we have also
perfect discrimination of two copies of measurements: $\poptent(\PP_{U^{\otimes
2}},\PP_{\1^{\otimes 2}})=1$. Therefore we have
\begin{equation}
\begin{split}
&\underset{U \sim \mu(\UU_d)}{\mathrm{Pr}}\left( \poptent (\PP_{U^{\otimes 
2}},\PP_{\1^{\otimes 
2}}) <1   \right) \\ &\leq 
\underset{U \sim \mu(\UU_d)}{\mathrm{Pr}}\left( \Vert \PP_{U} - \PP_\1 
\Vert_\diamond \leq \sqrt{2}   \right).
\end{split}
\end{equation}
Using now the characterization given in Eq.\eqref{eq:mainRESold} in conjunction 
with the formula in Eq.\eqref{eq:dimondNORMuni}  we obtain $\Vert \PP_{U} - 
P_\1 \Vert_\diamond \geq 2\sqrt{1-|U_{11}|^2}$ (note that 
$U_{11}=\tr(\ketbra{1}{1} 
U)\in W(U)$). Using this and simple algebra we get
\begin{equation}
\begin{split}
&\underset{U \sim \mu(\UU_d)}{\mathrm{Pr}}\left( \poptent ( \PP_{U^{\otimes
2}},\PP_{\1^{\otimes 2}}) <1   \right) \\ &\leq \underset{U \sim
\mu(\UU_d)}{\mathrm{Pr}}\left( |U_{11}|^2 \geq \frac{1}{2}  \right).
\end{split}
\end{equation}
The right-hand side of the above inequality can be computed exactly using the
property that for Haar-distributed $U$ the random variable $X=|U_{11}|^2$ is
distributed according to the beta distribution $p(X)=(d-1)(1-X)^{d-2}$ (see for
instance Eq.~(9) in~\cite{zyczkowski2000truncations}). The simple integration
gives $(1/2)^{d-1}$, which together with Eq.\eqref{eq:equalityOFprob} gives the
claimed result.
\end{proof}

\section{Unambiguous discrimination}\label{sec:unambiguousDISC} The unambiguous
discrimination of measurements $\PP_\1$ and $\PP_U$ can be understood as
unambiguous discrimination \cite{DIEKS1988} of states generated by the
corresponding channels. Specifically, for a fixed input state $\sigma$, the
output  states 
$\left(\PP_\1 \otimes \1\right) (\sigma), \left(\PP_U \otimes \1\right) 
(\sigma)$ 
can be unambiguously
discriminated using the measurement strategy $\MM=(M_\1, M_U, M_{?})$, where the
first two effects represent conclusive answers and  the last one corresponds to
the inconclusive
output of the procedure. For equal a priori probabilities of occurrence of
$\PP_\1$ and $\PP_U$, as well as fixed $\sigma$ (possibly entangled) and 
$\MM$,  the success
probability is given by
\begin{equation}\label{eq:unambSUCC}
\begin{split}
&p_\mathrm{u}\left(\PP_\1, \PP_U; \sigma, \MM \right) \\
&= \frac{1}{2} 
\tr(M_\1 (\PP_\1 \otimes \1) (\sigma))+ \frac{1}{2}\tr(M_U (\PP_U \otimes \1) 
(\sigma)),
\end{split}
\end{equation}
where additionally the unambiguity condition has to be satisfied:
\begin{equation}\label{eq:unambSTATEmes}
\tr(M_U (\PP_\1 \otimes \1) (\sigma))= \tr(M_\1 (\PP_U \otimes \1) (\sigma))=0.
\end{equation}

The optimal success probability of unambiguous discrimination of measurements
$\PP_\1,\PP_U$ can be now defined as the maximum of \eqref{eq:unambSUCC} over 
all strategies. Formally, we have
\begin{equation}
\begin{split}
&\pent\left(\PP_\1, \PP_U\right) \\ &\eqdef \max_{\sigma \in \DD(\C^d \otimes
\C^{d'})}\max_{\MM\in \mathrm{POVM}(\C^d \otimes \C^{d'})}
p_\mathrm{u}\left(\PP_\1, \PP_U; \sigma, \MM \right),
\end{split}
\end{equation}
where $\MM\in \mathrm{POVM}(\C^d \otimes \C^{d'})$ is a three-outcome 
measurement on $\C^d \otimes \C^{d'}$ 
that satisfies constrains \eqref{eq:unambSTATEmes} end $\sigma$ is a state on 
the extended Hilbert
space $\C^d \otimes \C^{d'}$. 

\subsection{Unambiguous discrimination with assistance of entanglement} 
In this subsection we present our main result which gives the
probability of unambiguous discrimination with the use of entanglement
This is presented as Theorem~\ref{thm:ENTunamb} while its proof is postponed to 
Appendix~\ref{app:ENTunamb}. Aside from giving a simple expression
for this probability, this result reduces the problem of unambiguous
measurement discrimination to a convex optimization task and gives a simple
relationship between the diamond norm and the probability of unambiguous
discrimination. 

\begin{theorem}\label{thm:ENTunamb}
The optimal success probability of unambiguous discrimination between von
Neumann measurements $\PP_\1$ and $\PP_U$ is given by
\begin{equation}
\pent(\PP_{\1},\PP_{U})=1-\min_{\rho \in \mathcal{D}(\mathbb{C}^d)}\sum_i 
|\bra{i} \rho U\ket{i}|.
\end{equation}
\end{theorem}
\begin{figure}[!h]
\includegraphics[width=.48\textwidth]{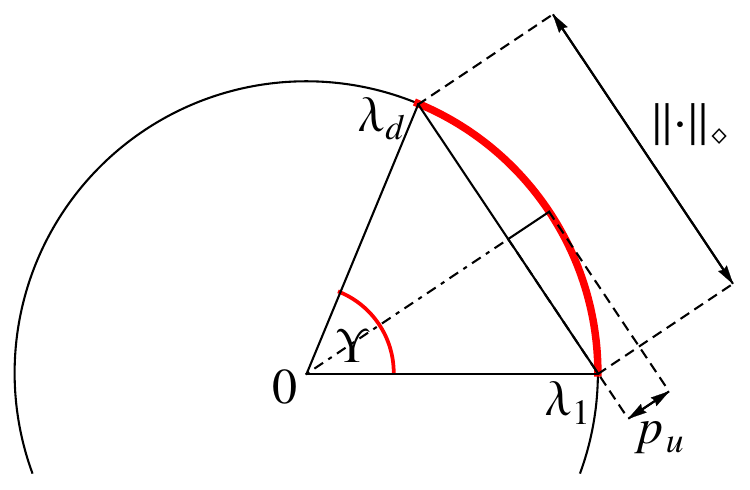} \caption{Schematic
depiction of the relationship between the diamond norm $\|
\PP_U-\PP_\1\|_\diamond$, the probability of unambiguous discrimination
$\pent(\PP_{\1},\PP_{U})$ and $\Upsilon(U)$.}\label{fig:geometric1}
\end{figure}
The results coming from Theorem~\ref{th:parallel} and 
Theorem~\ref{thm:ENTunamb} give a nice geometric 
interpretation for the relationship between
the diamond norm and the probability of unambiguous discrimination. This is
depicted in Fig.~\ref{fig:geometric1}. We start with a von Neumann measurement
in a basis given by some unitary matrix $U$ and try to distinguish it from the
measurement in the computational basis. We assume that $U$ is optimal and 
denote $U$'s eigenvalues as
$\lambda_1, \ldots, \lambda_d$ ordered according to their phases and use the 
symbol $\Upsilon(U)$ to denote the angle between two most distant eigenvalues 
$\lambda_1$ and $\lambda_d$. 
The dependence of the 
diamond norm and probability of unambiguous discrimination is clearly shown. 

\begin{remark}\label{rem:ua-parallel}
The above calculations can be easily extended to the case of parallel 
discrimination scheme. It suffices to substitute $U$ with $U^{\otimes N}$ and 
then we obtain that
\begin{equation}
\pent(\PP_{U^{\otimes N}}, \PP_\1) = 1 - \min_{\rho\in 
\mathcal{D}\left(\mathbb{C}^{d^N}\right)} \sum_{i} |\bra{i} \rho 
U^{\otimes N} 
\ket{i}|.
\end{equation}
\end{remark}

Basing on Remark~\ref{rem:ua-parallel}, we note that the angle $\Upsilon(U)$
increases in the multiple-shot case with the number of queries. This is 
depicted 
in Fig.~\ref{fig:geometric2} for two- and three-shots scenarios.
\begin{figure}[!h]
\centering
\begin{subfigure}{0.48\textwidth}
\includegraphics[width=\textwidth]{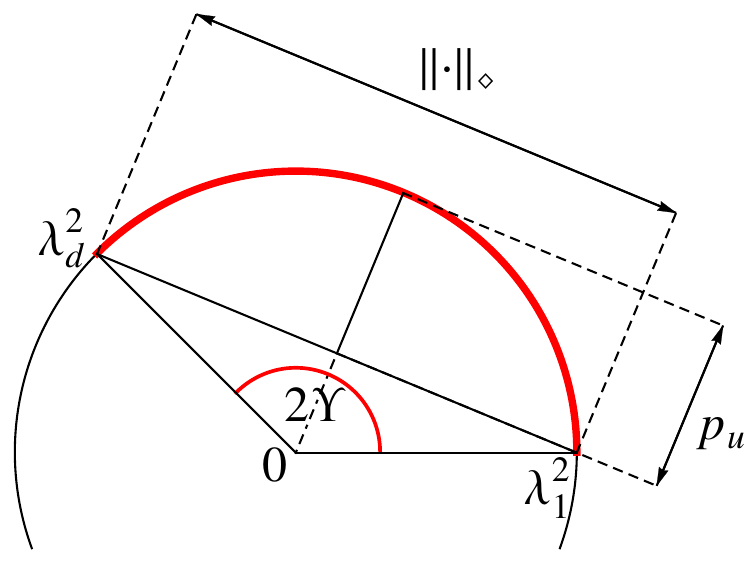}
\end{subfigure}
\begin{subfigure}{0.48\textwidth}
\includegraphics[width=\textwidth]{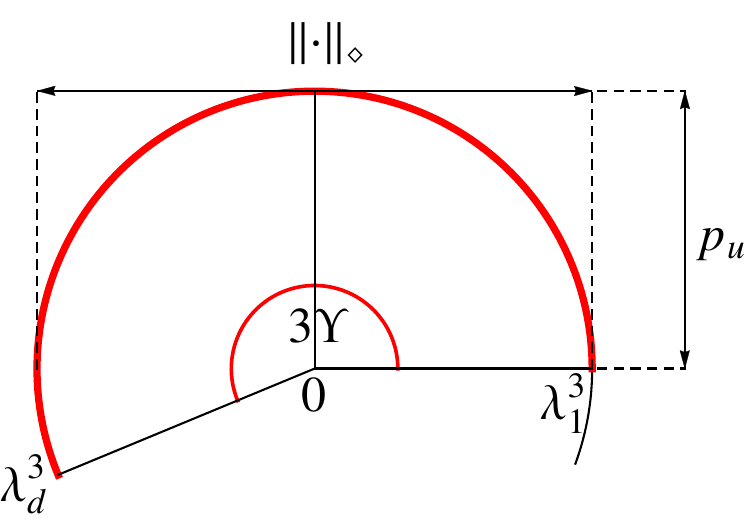}
\end{subfigure}
\caption{Figure similar to Fig.~\ref{fig:geometric1} which represents two- 
(left) and three- (right) shots scenario.} \label{fig:geometric2}
\end{figure}

In the general scheme we are
allowed to use conditional unitary transformations $\{V^{(k)}\} $ after each
measurement, thus our setting for discrimination is the same as presented in 
Fig.~\ref{fig:sequential_scheme} and Fig.~\ref{fig:au}. Similarly to the 
multiple-shot minimum error discrimination, we will show that adaptive 
discrimination scheme does not give any advantage 
over the parallel one for unambiguous discrimination. We state this formally 
in the 
following theorem, which proof is moved to Appendix~\ref{app:parallel-unam}.

\begin{theorem} \label{th:parallel-optimal-unamb}
Let $U\in \UU_d$.
Consider the unambiguous discrimination of general quantum network with $N$ 
uses of the 
black box in which 
there is one of two measurements - either $\PP_U$ or $\PP_\1$.
Then the probability of correct distinction cannot be better then
in the parallel scenario.
\end{theorem}

\subsection{Unambiguous discrimination without assistance of entanglement} 
In this subsection we provide a brief discussion on a special case of 
unambiguous discrimination without the utilization of entangled states.
We will use the following notation. 
Let $\Gamma, \Lambda \subset \{1,\ldots,d\}$. 
For given unitary matrix $U$, we define $P_\Gamma 
\coloneqq \sum_{i \in \Gamma} \proj{i}$ and $Q_\Lambda \coloneqq U P_\Lambda 
U^\dagger$. 
We set $\mathbb{P}_{\Gamma,\Lambda}$ 
to be the orthogonal projector onto $\textrm{Span}
\left(\{ U\ket{i}\}_{i \in \Gamma^c} \right) \cap \textrm{Span} \left(\{
\ket{j}\}_{j \in \Lambda^c} \right)$, where $\Gamma^c,\Lambda^c$ denote the 
complements of $\Gamma$ and $\Lambda$ respectively. 
The following theorem states the optimal success probability of unambiguous 
discrimination without the use of entanglement. 

\begin{theorem}\label{thm:entFREEunamb}
The optimal success probability of unambiguous discrimination, 
without the use of entanglement, between von Neumann measurements $\PP_\1$ and 
$\PP_U$ is given by
\begin{equation}
\begin{split}
&\p\left(\PP_\1, \PP_U\right) \\ &= \frac12 \max_{\Gamma, \Lambda 
\subset\{1,\dots ,
d\} : \Gamma \cap \Lambda = \emptyset} \Big\|\mathbb{P}_{\Gamma,\Lambda} (
P_\Gamma +Q_\Lambda)  \mathbb{P}_{\Gamma,\Lambda} \Big\|\,
\end{split}
\end{equation}
with $P_\Gamma,Q_\Lambda,\mathbb{P}_{\Gamma,\Lambda} $ defined as above.
\end{theorem}

The proof of this theorem is presented in Appendix~\ref{app:entFREEunamb}.

\begin{rem*}
The projector $\mathbb{P}_{\Gamma,\Lambda}$ projects onto the intersection of
supports of $P_{\Lambda^c}$ and $Q_{\Gamma^c}$. By the use of Theorem 4
from \cite{piziak1999constructing}, we can express the optimal probability of
unambiguous discrimination as
\begin{widetext}
\begin{equation}
\begin{split}
\p\left(\PP_\1, \PP_U\right) =  2\max_{\Gamma,\Lambda \subset\{1,\dots , 
d\} : 
\Gamma \cap 
\Lambda = \emptyset} 
\Big\|&
P_{\Lambda^c}
(P_{\Lambda^c}+Q_{\Gamma^c})^{-1}
Q_\Lambda 
(P_{\Lambda^c}+Q_{\Gamma^c})^{-1}\
P_{\Lambda^c}
\\&
+
Q_{\Gamma^c}
(P_{\Lambda^c}+Q_{\Gamma^c})^{-1}
P_\Gamma
(P_{\Lambda^c}+Q_{\Gamma^c})^{-1}\
Q_{\Gamma^c}
\Big\|,
\end{split}
\end{equation}
\end{widetext}
where $(\cdot)^{-1}$ denotes Moore-Penrose pseudo inverse~\cite{watroustheory}.
Moreover, the optimal input state is the one which gives the above norm.
\end{rem*}
\begin{corollary}
In the case of qubit measurements, the optimal probability of unambiguous
discrimination of $\PP_\1$ and $\PP_U$ is given by the
following (discontinuous) function
\begin{equation}
\p\left(\PP_\1, \PP_U\right) =
\begin{cases}
1 & \text{ if }  \  \ |U_{1,2}|^2 = 1\\
\frac12 |U_{1,2}|^2 & \text{ if }  \ \ |U_{1,2}|^2 < 1
\end{cases}\ .
\end{equation}
In both cases the optimal input state can be chosen to be $\ketbra{1}{1}$.
\end{corollary}

The following corollary states that in the qubit case the unambiguous 
discrimination with the assistance of entanglement always outperforms the 
unambiguous discrimination without the use of entanglement. 
On top of that, the special cases for which the use of entanglement does not 
give any advantage are described. 
\begin{corollary}
Let $\PP_\1$ and $\PP_U$ be two von Neumann measurements on a qubit. 
If $|U_{1,1}| \not\in  \{0, 1\} $, then 
the probability of entanglement-assisted unambiguous discrimination 
is given by
\begin{equation}
p_{\mathrm{u}} = 1- |U_{1,1}|
\end{equation} 
and it is always greater then the probability without assistance of 
entanglement
\begin{equation}
\p =
\begin{cases}
1 & \text{ if }  \  \ |U_{1,2}|^2 = 1\\
\frac12 |U_{1,2}|^2 & \text{ if }  \ \ |U_{1,2}|^2 < 1.
\end{cases}
\end{equation}
Moreover,  if $|U_{1,1}| \in  \{0, 1\}$, then $\pent = \p$.
\end{corollary}

\begin{rem*}
The above considerations can be extended to unambiguous discrimination of 
multiple copies of von Neumann measurements applied in parallel. To this end, 
if we have access to $N$ parallel queries to a black box measurement, it 
suffices to replace unitaries  $\1$ by  $\1^{\otimes N}$ and $U$  by 
$U^{\otimes N}$ in the above computations. Interestingly, in the contrast to 
unambiguous discrimination of quantum states \cite{DIEKS1988}, having access to 
two copies of black box measurement, sometimes allows attaining perfect 
discrimination. Specifically, consider the problem of discriminating between 
$\PP_\1$ and $\PP_H$, where $H=\frac{1}{\sqrt{2}}\begin{pmatrix} 
1 & 1 \\
1 & -1 
\end{pmatrix}$. Explicit computation shows that by taking the input state as 
$\ket{\Psi}=\frac{1}{\sqrt{2}}(\ket{1}\ket{1}-\ket{2}\ket{2})$ allows us to 
perfectly distinguish between $\PP^{\otimes 2}_\1$ and $\PP^{\otimes 2}_H$.  
\end{rem*}

\section{Conclusions}\label{sec:outro} We have presented a
comprehensive treatment of the problem of discrimination of von Neumann
measurements. First of all, we showed an alternative proof of the fact that for 
any pair of measurements $\PP_1$ and
$\PP_2$, $\PP_1 \neq \PP_2$, there exists a finite number $N$ of uses of the
black box which allows us to achieve perfect discrimination. Moreover, we 
calculated the exact value of the diamond norm for given $N$. This is formally
stated in Corollary~\ref{cor:finite-N}. We also proved that the parallel
discrimination scheme is optimal in the scenario of multiple-shot minimum error 
discrimination of von Neumann measurements (see 
Theorem~\ref{th:parallel-optimal}).

Moreover, we studied unambiguous discrimination of von Neumann measurements. Our
main contribution to this problem was the derivation of the general schemes that
attain the optimal success probability both with (see Theorem
\ref{thm:ENTunamb}) and without (see Theorem \ref{thm:entFREEunamb})  the
assistance of entanglement. Interestingly, for entanglement-assisted 
unambiguous discrimination the optimal success probability is functionally
related to the corresponding success probability for minimum error
discrimination. Finally, we show that the parallel scheme is also optimal for
unambiguous discrimination.

There are many interesting directions for further study that still
remain to be explored. Below we list the most important (in our opinion) open
research problems:
\begin{itemize}
\item Generalization of our results from projective measurements to other
classes of measurements such as projective-simulable measurements
\cite{Oszmaniec2017}, measurements with limited number of outcomes
\cite{Leo2017} or general quantum measurements (POVMs).

\item Systematical study of the problem of unambiguous discrimination of
projective measurements in the multiple-shot regime.

\item Can typical pairs of Haar-random projective measurements on $\C^d$  be
discriminated perfectly using only one query and the assistance of entanglement
as $d\rightarrow \infty$?

\item How much entanglement is needed to attain the optimal success
probability of multiple-shot discrimination of generic projective measurements
on $\C^d$? In the same scenario, is it necessary to adopt the final measurement
to the pair of measurements to be discriminated? 
\end{itemize}

\section*{Acknowledgements}
We would like to thank Alessandro Bisio and Giulio Chiribella for fruitful 
discussions.

This work was supported by the Polish National
Science Centre under project numbers 2016/22/E/ST6/00062 (ZP, AK, RK) and
2015/18/E/ST2/ 00327 ({\L}P).  MO acknowledges the support of Homing 
programme of the Foundation for Polish Science co-financed by the European 
Union under the European Regional Development Fund.

\bibliographystyle{ieeetr}
\bibliography{multiple_shot_measurement}

\onecolumn\newpage
\appendix

\section{Proof of Theorem \ref{th:parallel}}\label{app:parallel}
In this appendix we will begin with quoting two lemmas from 
\cite{puchala2018strategies} which contribute the main part of the proof of 
Theorem~\ref{th:parallel}. 
The first lemma states that a function $|\Tr(\rho UE)|$ has a saddle point.
The second lemma gives an equivalent condition to the existence of the saddle 
point and presents the form of the optimal state. 
Then, we present the proof of Theorem \ref{th:parallel} and a short discussion 
about the amount of entanglement needed for the discrimination.

\begin{lemma}[Lemma 4 from~\cite{puchala2018strategies}]\label{lem: state-1}
	Let us fix a unitary matrix $U \in \UU_d$. Then,
	\begin{equation}
	\max_{E \in \diaguni_d} \min_{\rho \in \DD(\C^d)}|\Tr(\rho UE)|=\min_{\rho 
	\in \DD(\C^d)}\max_{E \in \diaguni_d}|\Tr(\rho UE)|.
	\end{equation}
\end{lemma}
\begin{lemma}[Lemma 5 from~\cite{puchala2018strategies}]\label{lem: state}
	Let us fix a unitary matrix $U \in \UU_d$ and 
	\begin{itemize}
		\item $E_0 \in \diaguni_d$ and $D(E) = \min_{\rho\in 
			\DD(\C^d)} 
		|\Tr \rho UE|$,
		
		\item $D(E_0)>0$,
		
		\item $\lambda_1,\lambda_d$ denote the most distant pair of eigenvalues 
		of $UE_0$
		\item $P_1$, $P_d$ denote the projectors on the subspaces spanned by 
		the 
		eigenvectors corresponding to $\lambda_1$, $\lambda_d$.
	\end{itemize}
	Then, the function $ \DD(\C^d) \times  \diaguni_d\ni (\rho,E)\mapsto 
	|\Tr(\rho UE)|$ has a saddle point in $(\rho_0, E_0)$ i.e.
	\begin{equation}\label{app:eq_1}
	\max_{E \in \diaguni_d} \min_{\rho \in \DD(\C^d)}|\Tr(\rho UE)|=|\Tr(\rho_0 
	UE_0)|=\min_{\rho \in \DD(\C^d)}\max_{E \in \diaguni_d}|\Tr(\rho UE)|
	\end{equation} 
	if and only if there exist states $\rho_1, \rho_d$ such that
	\begin{itemize}
		\item $\rho_1=P_1 \rho_1 P_1$,
		\item $\rho_d=P_d \rho_d P_d$,
		\item $\diag (\rho_1)=\diag (\rho_d)$.
	\end{itemize}
Moreover, if the above holds, then the state $\rho_0$ satisfying 
Eq.~\eqref{app:eq_1} can be chosen as $\frac12 \rho_1 + \frac12 \rho_d$.
\end{lemma}

\begin{proof}[Proof of Theorem \ref{th:parallel}]
In the first step of the proof assume that $\PP_U$ and $\PP_\1$ are perfectly
distinguishable  in a single-shot scenario i.e.
\begin{equation}
\left\Vert  \PP_U - \PP_\1 \right\Vert_\diamond = \min_{E \in 
	\mathcal{DU}_d}
\left\Vert \Phi_{UE} - \Phi_\1 \right\Vert_\diamond = 2.
\end{equation}
This trivially implies that for each $N \in \mathbb{N}$ the measurements
$\PP_{U^{\otimes N}}$ and $\PP_{\1}$ are perfectly distinguishable and it holds
that
\begin{equation}
\begin{split}
2&=\left\Vert  \PP_{U^{\otimes N}} - \PP_{\1} \right\Vert_\diamond = 
\min_{F \in \mathcal{DU}_{d^N}} \left\Vert  \Phi_{(U^{\otimes N})F} - \Phi_{\1} 
\right\Vert_\diamond 
\\
&\leq \min_{E \in \mathcal{DU}_d} \left\Vert  
\Phi_{U^{\otimes N} E^{\otimes N}} - \Phi_{\1} 
\right\Vert_\diamond \leq 2,
\end{split}
\end{equation}
which proves the thesis of the theorem in this case.

Now, consider the second case when $\PP_U$ and $\PP_\1$ are not perfectly
distinguishable using a single query. Then, according to equality
\begin{equation}
\left\Vert \PP_U - \PP_\1 \right\Vert_\diamond = \min_{E \in \mathcal{DU}_d} 
\left\Vert \Phi_{UE} - \Phi_\1 \right\Vert_\diamond 
\end{equation}
there exists an optimal matrix $E_0 \in \mathcal{DU}_d$
such that $0 \not\in W(UE_0)$ and $\left\Vert  \PP_U - \PP_\1 
\right\Vert_\diamond = \left\Vert
\Phi_{UE_0} - \Phi_\1 \right\Vert_\diamond$.

The general proof strategy is to utilize the Lemma \ref{lem: state} in order to 
construct optimal discriminator
$\rho_0$ for measurements $\PP_{U^{\otimes N}}$ and $\PP_{\1}$. Hence we start 
by establishing that the assumptions of this lemma are fulfilled. This will 
follow from Lemma~\ref{lem: state-1}. According to it the
function $(\rho, E) \mapsto |\Tr(\rho UE)|$ has a saddle point. Let us remind 
that $\left\Vert \Phi_U - \Phi_\1 \right\Vert_\diamond = 2 \sqrt{1 - \nu^2}$, 
where $\nu  =  \min_{x \in W(U)} \vert x \vert $. Due to equality
$\min_{E \in \mathcal{DU}_d} \left\Vert \Phi_{UE} - \Phi_{\1}
\right\Vert_\diamond = \left\Vert \Phi_{UE_0} - \Phi_\1 \right\Vert_\diamond$ 
we have 
\begin{equation}
\max_{E \in \DU_d} \min_{x \in W(UE)} |x| = \min_{x \in W(UE_0)} |x|.
\end{equation}
Hence using the property that the numerical range is a convex set 
\cite{hausdorff1919wertvorrat,toeplitz1918algebraische} we obtain
\begin{equation}
\max_{E \in \DU_d} \min_{\rho \in \DD(\C^d)}|\Tr(\rho 
UE)|=\min_{\rho \in \DD(\C^d)}|\Tr(\rho 
UE_0)|.
\end{equation} 
% $(\rho_0,E_0)$,
Therefore the former assumptions of Lemma \ref{lem: state} are satisfied for 
the matrix
$E_0$ and hence it is possible to take states $\rho_1$, $\rho_d$ which fulfill
the latter conditions of this Lemma.

To complete the proof we will separately study two cases. In the first one, we
assume that POVMs $\PP_{U}$ and $\PP_{\1}$ are not perfectly distinguishable 
using $N$
queries i.e.  $0 \not\in W 
\left(U^{\otimes N}
E_0^{\otimes N} \right)$. Hence, as $\diag (\rho_1) = \diag (\rho_d)$, then  
\begin{equation}
\diag (\rho_1^{\otimes N}) 
= \diag (\rho_d^{\otimes N}) 
\end{equation}
and $\rho_1^{\otimes N},\rho_d^{\otimes N}$ lie on the subspaces spanned by the
eigenvectors of the matrix $U^{\otimes N} E_0^{\otimes N}$ eigenvalues
$\lambda_1^N$ and $\lambda_d^N$, respectively. Consequently, we fulfilled
the latter assumptions of Lemma \ref{lem: state} and the reverse implication of 
this Lemma
states that the unitary matrix $E_0^{\otimes N}$ is optimal and for
$\rho_0=\frac12 \rho_1^{\otimes N} + \frac12 \rho_d^{\otimes N}$ it holds that
\begin{equation}
\begin{split}
\min_{\rho \in \DD(\C^d)} \left\vert \tr \left( \rho (UE_0)^{\otimes N} 
\right) \right\vert
&=  \left\vert \tr \left( \rho_0 (UE_0)^{\otimes N} 
\right) \right\vert \\ &= \max_{F \in \mathcal{DU}_{d^N}} 
\min_{\rho \in \DD(\C^{d^N})} \left\vert \tr \left( \rho U^{\otimes N} F 
\right) \right\vert.
\end{split}
\end{equation}
Hence 
\begin{equation}
\left\Vert  \Phi_{U^{\otimes N}
E_0^{\otimes N}} - 
\Phi_{\1} \right\Vert_\diamond  
= \min_{F \in \mathcal{DU}_{d^N}} \left\Vert  
\Phi_{U^{\otimes N}F} 
- \Phi_{\1} \right\Vert_\diamond \leq \min_{E \in \mathcal{DU}_d} \left\Vert  
\Phi_{U^{\otimes N}
	E^{\otimes N}} - \Phi_{\1} 
\right\Vert_\diamond 
\end{equation}
and eventually
\begin{equation}
 \left\Vert  \PP_{U^{\otimes N}} - \PP_{\1} \right\Vert_\diamond=\min_{E \in 
 \mathcal{DU}_d} \left\Vert  
\Phi_{U^{\otimes N}
E^{\otimes N}} - \Phi_{\1} 
\right\Vert_\diamond. 
\end{equation}

In the second case, let  
$0 \in W\left( U^{\otimes M} E_0^{\otimes M} \right)$. Let us consider the
situation when $M$ is the first index for which this happens,  which means that
\linebreak $0 \not\in W\left( U^{\otimes M-1} E_0^{\otimes M-1} \right)$.  

If the measurements in question can be perfectly distinguished with $M$ queries,
then we have $0 \in \textrm{conv} (\lambda_1^M, \lambda_1 \lambda_d^{M-1},
\lambda_d^M)$ and there exists a probability vector $p = (p_1, p_2, p_3)$ such
that
\begin{equation}
p_1 \lambda_1^M + p_2 \lambda_1 \lambda_d^{M-1} + p_3 \lambda_d^M = 0.
\end{equation}
Define a state
\begin{equation}\label{pf-eq-1}
\rho = p_1 \rho_1^{\otimes M}
+ p_2 \left(\rho_1 \otimes \rho_d^{\otimes M-1} \right) 
+ p_3 \rho_d^{\otimes M}.
\end{equation}
We will show that $\diag \left(\rho U^{\otimes M}\right) = 0$. Indeed
\begin{equation}\label{pf-eq-2}
\begin{split}
&\diag \left(\rho U^{\otimes M} 
E_0^{\otimes M} \right) \\ 
=& \diag \left(p_1 \lambda_1^M \rho_1^{\otimes M}  +
p_2  \lambda_1 \lambda_d^{M-1} \left(\rho_1 \otimes  \rho_d^{\otimes M-1}
\right)  
+ p_3 \lambda_d^M \rho_d^{\otimes M} \right) \\ 
=& p_1 \lambda_1^M \diag (\rho_1^{\otimes M})
+ p_2 \lambda_1 \lambda_d^{M-1} \diag \left(\rho_1 \otimes \rho_d^{\otimes 
M-1} \right)
+ p_3 \lambda_d^M \diag (\rho_d^{\otimes M}) \\
=& \left( p_1 \lambda_1^M + p_2  \lambda_1 \lambda_d^{M-1} + p_3 \lambda_d^M 
\right) \diag (\rho_1^{\otimes M}) =0.
\end{split}
\end{equation}
Thus from Proposition 3 form \cite{puchala2018strategies} the condition $\diag 
\left(\rho U^{\otimes M} 
E_0^{\otimes M} \right)=0$ 
implies that $\left\Vert  \PP_{U^{\otimes M}} - \PP_{\1} \right\Vert_\diamond = 
2$, and hence
\begin{equation}
\left\Vert  \PP_{U^{\otimes M}} - \PP_{\1} \right\Vert_\diamond = \min_{E \in 
\mathcal{DU}_d} \left\Vert  
\Phi_{U^{\otimes M} E^{\otimes M}} - 
\Phi_{\1} \right\Vert_\diamond. 
\end{equation}
For $N>M$, the equality $\left\Vert  \PP_{U^{\otimes M}} - \PP_{\1} 
\right\Vert_\diamond=2$ implies that $\left\Vert  \PP_{U^{\otimes N}} - 
\PP_{\1} \right\Vert_\diamond=2$. Therefore
\begin{equation}
\left\Vert  \PP_{U^{\otimes N}} - \PP_{\1} \right\Vert_\diamond=\min_{E \in 
\mathcal{DU}_d} \left\Vert  
\Phi_{U^{\otimes N} E^{\otimes N}} - 
\Phi_{\1} \right\Vert_\diamond, 
\end{equation}
which completes the proof.
\end{proof}

Let us discuss the amount of the discriminator's entanglement with environment. 
Note
that the minimal dimension of an auxiliary system needed for optimal
discrimination is equal to the rank of the state $\rho_0$. One special case
involves the situation $\Vert \PP_{U^{\otimes N}} - \PP_\1 \Vert_\diamond<2$,
where $\rho_0=\frac12 \rho_1^{\otimes N} + \frac12 \rho_d^{\otimes N}$. In the
best case, when eigenvalues $\lambda_1$ and $\lambda_d$ are not degenerated, the
states $\rho_1, \rho_d$ are pure, hence $\text{rank} (\rho_0)=2$. On the other
hand, when the spectrum of $UE_0$ contains only $\lambda_1$ and $\lambda_d$,
then the rank of $\rho_0$ can be roughly upper-bounded by $(d-1)^N+1$.

The situation $\Vert \PP_{U^{\otimes N}} -\PP_\1 \Vert_\diamond=2$ is more 
complicated to analyze, due to lack of an analytic form of the discriminator in 
a general case. However, finding a pair of not degenerated eigenvalues 
$\lambda_1,\lambda_d$ 
saturating the latter assumptions of Lemma~\ref{lem: state} will lead to 
$\text{rank} (\rho_0)=3$, where $\rho_0$ is defined as in Eq.~\eqref{pf-eq-1} 
and we check its optimality in the same manner as in Eq.~\eqref{pf-eq-2}.

\section{Explicit form of the matrix $A_U$}\label{app:au}
 The matrix $A_U$ is a general operation which allows for adaptive information 
processing in the sequential discrimination scenario of von Neumann 
measurements. It consists of a sequence of unitary matrices $U$ 
acting successively on given registers interlacing with classically controlled 
unitary operations $V^i$. The explicit form of the matrix $A_U$ is given by
\begin{equation}
\begin{split}
A_U = 
&\left(  \1_{1, \ldots N-1} \otimes U \otimes \1_{N+1}  \right)\\
&\left(
\sum_{i_1,\ldots,i_{N-1}} \proj{i_1,
	\ldots,i_{N-1}} \otimes 
V^{(N-1)}_{i_1,
	\ldots,i_{N-1}}  \right)  \\
&
\left( \1_{1, \ldots , N-2} \otimes U \otimes \1_{N, N+1}   \right)\\
&
\left(
\sum_{i_1,\ldots,i_{N-2}} \proj{i_1,
	\ldots,i_{N-2}} \otimes 
V^{(N-2)}_{i_1,
	\ldots,i_{N-2}}  \right) \\
&\ldots
\\
&\left(  \1_1 \otimes U \otimes \1_{3, \ldots N+1}  \right)\\
&\left( \sum_{i_1} \proj{i_1} \otimes  V^{(1)}_{i_1} \right)\\
&\left( U \otimes \1_{2, \ldots N+1} \right).
\end{split}
\end{equation}

\section{Proof of Theorem \ref{thm:ENTunamb}}\label{app:ENTunamb}

\begin{proof}[Proof of Theorem \ref{thm:ENTunamb}]
 Assume that  a black-box measurement ($\PP_\1$ or $\PP_U$) acts on a system 
 extended by the ancilla space $\HH_B$  (of some dimension $d_1$). Without loss 
 of generality we take the pure input state i.e. 
 $\sigma=\ketbra{\psi_{AB}}{\psi_{AB}}$. Let $X$ be a matrix such
that $\ket{\psi_{AB}}=\sum_{i,j=1}^{d,d_1} X_{ij}\ket{i}\ket{j}$. The action of
the channels $\PP_\1 \otimes \1_B$ and $\PP_U \otimes \1_B$ on 
$\ketbra{\psi_{AB}}{\psi_{AB}}$ can be expressed as
\begin{equation}\label{eq:postSTSATE}
\begin{split}
\left(\PP_\1 \otimes \1_B  \right) 
(\ketbra{\psi_{AB}}{\psi_{AB}})&=\sum_{i=1}^d 
\ketbra{i}{i}\otimes X^T \ketbra{i}{i} \overline{X},\\
\left(\PP_U \otimes \1_B  \right)(\ketbra{\psi_{AB}}{\psi_{AB}})&=\sum_{i=1}^d 
\ketbra{i}{i}\otimes X^T \overline{U} \ketbra{i}{i}U^T \overline{X},
\end{split}
\end{equation}
where we treat measurements $\PP_\1$ and $\PP_U$ as a measure and prepare 
channels of the form $\Psi_\MM (\sigma) =\sum_{i=1}^n \tr(M_i \sigma) 
\ketbra{i}{i}$.

As for any Hermitian operator $M$ and any measurement $\PP$ we have
\begin{equation}\label{eq:dephasing}
\begin{split}
\tr \left( M (\PP \otimes \1)(\sigma)  \right)
&=\tr \left( M (\Delta\PP \otimes \1)(\sigma) \right)  \\
&=\tr \left( (\Delta \otimes \1)(M) (\PP \otimes \1)(\sigma) \right),
\end{split}
\end{equation}
where $\Delta$ is a dephasing channel. Hence we can restrict our attention to 
considering measurements $\MM$ which effects have block-diagonal structure, 
that is 
%As presented in 
%Eq.~\eqref{eq:dephasing} the optimal 
%measurement unambiguously discriminating the above two states has the 
%following 
%structure
\begin{equation}
\MM = \sum_{i=1}^d \ketbra{i}{i}\otimes \TT_i, 
\end{equation}
where $\TT_i$ is a POVM on $\HH_B$ associated with a measure and prepare 
channel. From 
Eq.\eqref{eq:postSTSATE} we see that upon obtaining the label 
$i$, the state of the auxiliary subsystem is either
\begin{equation}
\ketbra{x_i}{x_i} = p_i^{-1} X^\top \ketbra{i}{i} \overline{X},
\end{equation}
when measurement $\PP_\1$ was performed, or it is given by
\begin{equation}
\ketbra{y_i}{y_i} = q_i^{-1} X^\top \overline{U}\ketbra{i}{i}U^T \overline{X}
\end{equation} 
if $\PP_U$ was implemented. In the above formulas $p_i,q_i$ are responsible for 
normalization. We assume that $p_i>0$ and $q_i>0$ (otherwise the specific 
outcome $i$ does not occur). We see that states  $\ketbra{x_i}{x_i}$, 
$\ketbra{y_i}{y_i}$ are pure and therefore the optimal measurements 
$\TT_i=\lbrace T^{(i)}_1, T^{(i)}_2, T^{(i)}_? \rbrace$ 
will be simply given by
\begin{equation}
\begin{split}
T^{(i)}_1&=\gamma_1(\1-\proj{y_i}),\\
T^{(i)}_2&=\gamma_2(\1-\proj{x_i}),\\
T^{(i)}_?&=\1-T_1-T_2,
\end{split}
\end{equation}
for some choice of $\gamma_{1,2}$ which guarantees the non-negativity of
$T^{(i)}_?$. 

The probability of success in unambiguous discrimination of pure states
$\ket{x}, \ket{y}$ with unequal a priori probabilities $\eta,1-\eta$ is given 
by~\cite{jaeger1995optimal}
\begin{equation}
p^u_{succ}(x,y,\eta) =
\begin{cases}
    1 -\eta - (1-\eta) c^2 & \text{for } \eta < \frac{c^2}{1+c^2} \\
    1 - 2 c \sqrt{\eta(1-\eta) } & \text{for } 
    \frac{c^2}{1+c^2} \leq \eta \leq \frac{1}{1+c^2}\\
    1 -(1-\eta) - \eta c^2 & \text{for } \frac{1}{1+c^2} < \eta,
  \end{cases}
\end{equation}
where $c = |\braket{x}{y}|$.

We will use the following upper bound
\begin{equation}
p^u_{succ}(x,y,\eta) \leq 1 - 2 c \sqrt{\eta(1-\eta)},
\end{equation}
which can be verified directly by elementary calculations.

Let $\rho=X X^\dagger$. The overlap $c_i$ between states of the auxiliary 
subsystem is given by
\begin{equation}
c_i = |\braket{x_i}{y_i}| =  |\bra{i} \overline{X} X^\top 
\overline{U}\ket{i}|/\sqrt{p_i 
q_i} = |\bra{i} \rho U\ket{i}| / \sqrt{p_i q_i},
\end{equation}
while a priori probabilities of $\ket{x_i}, \ket{y_i}$ upon obtaining label $i$ 
are $\eta_i = \frac{p_i}{p_i+q_i}, 1- \eta_i = \frac{q_i}{p_i+q_i}$.
Taking the above into account, we get that probability of success in 
unambiguous measurement on auxiliary subsystem, given that label $i$ was 
observed, can be bounded from above by
\begin{equation}\label{key}
p^u_{succ}(x_i,y_i,\eta_i) \leq 1 - 2 c_i \frac{\sqrt{p_i q_i}}{p_i+q_i}
= 1 - \frac{2 |\bra{i} \rho U\ket{i}|}{p_i+q_i}.
\end{equation}
Therefore, the overall probability of success is bounded by
\begin{equation}
\begin{split}
\pent(\PP_{\1},\PP_{U}) &= 
\max_{\ket{\psi_{AB}}} \sum_i \mathrm{Pr}(\text{label} = i) 
p^u_{succ}(x_i,y_i,\eta_i)  \\
&\leq
\max_{\rho} \sum_i \frac12 (p_i + q_i)\left(1 - \frac{2 |\bra{i} 
\rho U\ket{i}|}{p_i+q_i}\right)\\
&=1-\min_{\rho}\sum_i |\bra{i} \rho U\ket{i}|.
\end{split}
\end{equation}

What is more, the above bound is tight. The situation when $\PP_\1$ and $\PP_U$ 
are perfectly distinguishable is trivial to check. In the case when $\PP_\1$ 
and $\PP_U$ are not perfectly distinguishable, then there exists a  
state $\rho$ which will give equal probabilities $p_i$ and $q_i$ for each label 
$i$. This statement follows from Lemma~\ref{lem: state}, from which we take 
$\rho = \rho_0$.
\end{proof}

\section{Proof of 
Theorem~\ref{th:parallel-optimal-unamb}}\label{app:parallel-unam}

\begin{proof}[Proof of Theorem \ref{th:parallel-optimal-unamb}]

We will assume that the unitary matrix $U$ is
optimal, \ie  $\Upsilon(U)=\Theta(U)$. Take an arbitrary input state $\sigma = 
\proj{\psi_{A,B}}$. Let us denote
\begin{equation}
\begin{split}
&\ket{x_i} = p_i^{-1/2} \left( \bra{i} \otimes \1_{N+1} \right) A_{\1} 
\ket{\psi_{A,B}} \\
&\ket{y_i} = q_i^{-1/2} \left( \bra{i} \otimes \1_{N+1} \right) A_{U} 
\ket{\psi_{A,B}},
\end{split}
\end{equation}
where $A_{U}$ and $A_{\1}$ are defined as in Appendix~\ref{app:au} and $p_i, 
q_i$ are responsible for normalization. Repeating the calculation from the 
single-shot scenario from the proof in Appendix \ref{app:ENTunamb} we can 
upper-bound the probability of successful discrimination as follows
\begin{equation}
\begin{split}
\pent(\Psi_U,\Psi_\1) & \leq  1- \min_{\ket{\psi_{A,B}}} \sum_i 
\left| \bra{\psi_{A,B}} A_\1^\dagger \left(\proj{i} \otimes \1_{N+1}\right) A_U 
\ket{\psi_{A,B}} \right| \\
& \leq 1-\min_{\ket{\psi_{A,B}}} \left| \sum_i \bra{\psi_{A,B}} A_\1^\dagger 
\left( \proj{i} 
\otimes \1_{N+1} \right) A_U 
\ket{\psi_{A,B}} \right| \\
& = 
1- \min_{\ket{\psi_{A,B}}} \left| \bra{\psi_{A,B}} A_\1^\dagger
A_U \ket{\psi_{A,B}} \right|.
\end{split}
\end{equation}
From the work \cite{chiribella2008memory} 
we know that there exists a state
$\ket{\phi}$ such that for all $\ket{\psi_{A,B}}$ it holds that
\begin{equation}
|\bra{\psi_{A,B}} A_\1^\dagger A_U 
\ket{\psi_{A,B}}| \geq |\bra{\phi} U^{\otimes N} \ket{\phi}|.
\end{equation}
Moreover, using optimality of $U$ and Lemma~\ref{lem: state-1}, the state 
$\ket{\phi}$ 
can be chosen to satisfy 
$|\bra{\phi} U^{\otimes N} 
\ket{\phi}| = \min_{\rho} \sum_{i} |\bra{i} \rho
U^{\otimes N} \ket{i}|$. 
This leads to the desired inequality
\begin{equation}
\pent(\Psi_U,\Psi_\1) \leq  1- \min_{\rho} \sum_{i} |\bra{i} \rho
U^{\otimes N} \ket{i}|.
\end{equation}
\end{proof}

\section{Proof of 
Theorem~\ref{thm:entFREEunamb}}\label{app:entFREEunamb}

\begin{proof}[Proof of Theorem \ref{thm:entFREEunamb}]

Let us fix the input state $\sigma \in \DD(\C^d)$.
Without loss of generality we can
restrict our attention to measurements $\MM$ with diagonal effects
(see Eq.~\eqref{eq:dephasing} in Appendix \ref{app:ENTunamb}).

From the unambiguity condition we obtain  
that $M_\1
\perp \mathrm{supp}(\PP_U (\sigma))$ and $M_U \perp \mathrm{supp}(\PP_\1
(\sigma))$. Therefore, the optimal measurements can be always chosen as
projectors onto disjoint subsets $\Gamma, \Lambda$ of $\lbrace 1,\ldots,d
\rbrace$. The formula for the success probability 
reads
\begin{equation} \label{eq:p_unambiguous}
\p\left(\PP_\1,\PP_U; \sigma, \Gamma,\Lambda \right) = \frac{1}{2} 
\tr(P_\Gamma
\sigma) + \frac{1}{2} \tr(Q_\Lambda \sigma).
\end{equation}
Importantly, the input state $\sigma$ satisfies $\sigma\perp P_{\Lambda}$ and
$\sigma\perp Q_{\Gamma}$.  For
fixed subsets $\Lambda,\Gamma$, due to linearity, the maximum over $\sigma$
equals $\|\mathbb{P}_{\Gamma,\Delta} ( P_\Gamma +Q_\Lambda) 
\mathbb{P}_{\Gamma,\Delta} \|$, where $\|\cdot\|$ denotes the operator norm and
$\mathbb{P}_{\Gamma,\Delta}$ is the orthogonal projector onto $\textrm{Span}
\left(\{ U\ket{i}\}_{i \in \Gamma^c} \right) \cap \textrm{Span} \left(\{
\ket{j}\}_{j \in \Lambda^c} \right)$. By optimizing over disjoint subsets
$\Lambda,\Gamma\subset\lbrace 1,\ldots,d\rbrace $ we obtain the  result.

\end{proof}
\end{document}